\documentclass[11pt]{article}
\usepackage{amsmath}
\usepackage{amssymb}
\usepackage{amsfonts}
\usepackage{array}
\usepackage{float}
\usepackage{enumerate}
\usepackage{a4wide}
\usepackage{ifthen}
\usepackage{xspace}
\usepackage{graphicx}
\usepackage{arydshln}
\usepackage{subfigure}

\usepackage{algorithmic}
\usepackage{algorithm}

\floatstyle{ruled}
\newfloat{algorithm}{thp}{lop}[section]
\floatname{algorithm}{Algorithm}
\newfloat{module}{thp}{lop}[section]
\floatname{module}{Module}

\newtheorem{thm}{Theorem}[section]

\newtheorem{lem}[thm]{Lemma}
\newtheorem{definition}{Definition}[section]

\newenvironment{proof}{\emph{Proof}:}{\hfill$\Box$}

\newcommand{\remove}[1]{}

\newenvironment{th-repeat}[1]{\begin{trivlist}
\item[\hspace{\labelsep}{\bf\noindent Theorem~\ref{#1} }]}%
{\end{trivlist}}

\newenvironment{lem-repeat}[1]{\begin{trivlist}
\item[\hspace{\labelsep}{\bf\noindent Lemma~\ref{#1} }]}%
{\end{trivlist}}

\begin{document}


\title{Byzantine Convergence in Robots Networks:\\ The Price of Asynchrony} 

\author{Zohir Bouzid \and Maria Gradinariu \and S\'{e}bastien Tixeuil}

\date{Universit\'e Pierre et Marie Curie - Paris 6, LIP6-CNRS 7606, France\\
\texttt{FirstName.LastName@lip6.fr}}

\maketitle

\begin{abstract}
We study the convergence problem in fully asynchronous, uni-di\-men\-sio\-nal robot networks that are prone to Byzantine (\emph{i.e.} malicious) failures.
In these settings, oblivious anonymous robots with arbitrary initial positions are required to eventually 
converge to an \emph{a apriori} unknown position despite a subset of them exhibiting Byzantine behavior.  
Our contribution is twofold. We propose a deterministic algorithm that solves the problem in 
the most generic settings: fully asynchronous robots that operate in the non-atomic CORDA model.
Our algorithm provides convergence in $5f+1$-sized networks where $f$ is the upper bound 
on the number of Byzantine robots. Additionally, we prove that $5f+1$ is a lower bound whenever robot scheduling is fully asynchronous. 
This constrasts with previous results in partially synchronous robots networks, where $3f+1$ robots are necessary and sufficient.


\end{abstract}
\section{Introduction} 
The use of cooperative swarms of weak inexpensive robots for achieving 
complex tasks such as exploration or tracking in dangerous environments is
a promising option for reducing both human and material costs. 
Robot networks recently became a challenging research area 
for distributed systems since most of the problems to be solved in
this context (\emph{e.g.} coordination, agreement, resource allocation or 
leader election) form the core of distributed computing. However,
the classical distributed computing solutions do not translate well due to
fundamentally different execution models.

In order to capture the essence of distributed coordination in robot
networks, two main computational models are proposed in the
literature: the ATOM~\cite{SY99} and CORDA~\cite{Pre05} models. The
main difference between the two models comes from the granularity for
executing a \emph{Look-Compute-Move} cycle. In such a cycle, the Look phase
consists in taking a snapshot of the other robots positions using its
visibility sensors. In the Compute phase a robot computes a target
destination based on its previous observation. The Move phase simply
consists in moving toward the computed destination using motion
actuators. In the ATOM model, the whole cycle is atomic while in the CORDA model, the cycle is executed in a continuous manner. That is, in the ATOM model, robots
executing concurrently always remain in the same phase while in CORDA it is possible that \emph{e.g.} a robot executes its Look phase while another robot performs its
Move phase, or that a robot executes its Compute phase while its view (obtained during the Look phase) is already outdated. Of course, executions that may appear in the CORDA model are a strict superset of those that may appear in the ATOM model, so a protocol that performs in the CORDA model also works in the ATOM model, but the converse is not true. Similarly, impossibility results for the ATOM model still hold in the CORDA model. Complementary to the granularity of robots action is the amount of \emph{asynchrony} in the system, that is modeled by the scheduler: \emph{(i)} a \emph{fully synchronous} scheduler operates all robots in a lock-step manner forever, while \emph{(ii)} a $k$-bounded scheduler preserves a ratio of $k$ between the most often activated robot and the least often activated robot, finally \emph{(iii)} a \emph{fully asynchronous} scheduler only guarantees that every robots is activated infinitely often in an infinite execution.
The robots that we consider have weak capacities: they are \emph{anonymous} (they execute the same protocol and have no mean to distinguish themselves from the others), \emph{oblivious} (they have no memory that is persistent between two cycles), and have no compass whatsoever (they are unable to agree on a common direction or orientation).

\emph{Convergence} is a fundamental agreement primitive in robot networks and is used in the implementation 
of a broad class of services (\emph{e.g.} the construction of common coordinate systems or specific geometrical patterns). 
Given a set of oblivious robots with arbitrary initial locations and
no agreement on a global coordinate system, \emph{convergence}
requires that all robots asymptotically approach the same, but unknown
beforehand, location. Convergence looks similar to distributed approximate agreement since both problems require nodes to agree on
a common object (that is instantiated to be a position in space for the case of convergence, or a value in the case of
distributed agreement). 



\paragraph{Related works}

Since the pioneering work of Suzuki and Yamashita~\cite{SY99}, gathering\footnote{Gathering requires  robots to actually \emph{reach} a single point within finite time regardless of their initial positions.} and convergence have been addressed in \emph{fault-free} systems for a broad class of settings. Prencipe~\cite{Pre05} studied the problem of gathering in both ATOM and CORDA models, and showed that the problem is intractable without additional assumptions such as being able to detect the multiplicity of a location (\emph{i.e.}, knowing if there is more than one robot in a given location). 

The case of \emph{fault-prone} robot networks was recently tackled by several academic studies. The faults that have been investigated fall in two categories: \emph{crash} faults (\emph{i.e.} a faulty robots stops executing its cycle forever) and \emph{Byzantine} faults (\emph{i.e.} a faulty robot may exhibit arbitrary behavior and movement). Of course, the Byzantine fault model encompasses the crash fault model, and is thus harder to address.
\emph{Deterministic} fault-tolerant gathering is addressed in~\cite{agmon2004ftg} where the authors study a gathering protocol that tolerates one crash, and an algorithm for the ATOM model with fully synchronous scheduling that tolerates up to $f$ byzantine faults, when the number of robots is (strictly) greater than $3f$. In \cite{defago3274fta} the authors study the feasibility of \emph{probabilistic} gathering in crash-prone and Byzantine-prone environments.
\emph{Deterministic} fault-tolerant convergence was first addressed in~\cite{cohen2004rcv,cohen2005cpg}, where algorithms based on convergence to the center of gravity of the system are presented. Those algorithms work in the ATOM~\cite{cohen2004rcv} and CORDA~\cite{cohen2005cpg} models with a fully asynchronous scheduler and tolerate up to $f$ ($n>f$) crash faults, where $n$ is the number of robots in the system. Most related to this paper are~\cite{BGT09c,BGT09d}, where the authors studied convergence in byzantine-prone environments when robots move in a uni-dimensional space. In more details, \cite{BGT09c} showed that convergence is impossible if robots are not endowed with strong multiplicity detectors which are able to detect the exact number of robots that may simultaneously share the same location. The same paper defines the class of \emph{cautious} algorithms which guarantee that correct robots always move inside the range of positions held by correct robots, and proved that any cautious convergence algorithm that can tolerate $f$ Byzantine robots requires the presence of at least $2f+1$ robots in fully-synchronous ATOM networks and $3f+1$ robots in $k$-bounded (and thus also in fully asynchronous) ATOM networks. The lower bound for the ATOM model naturally extends to the CORDA model, and~\cite{BGT09d} provides a matching upper bound in the $k$-bounded CORDA model.

Interestingly enough, all previously known deterministic Byzantine tolerant robot protocols assume either the more restrictive ATOM model~\cite{defago3274fta}, or the constrained fully synchronous~\cite{agmon2004ftg} or $k$-bounded~\cite{BGT09c,BGT09d} schedulers, thus the question of the existence of such protocols in a fully asynchronous CORDA model remains open.

\paragraph{Our contribution}
We present the first study of Byzantine resilient robot protocols that considers the most general execution model: the CORDA model together with the fully asynchronous scheduler. We concentrate on the convergence problem and prove that the fully asynchronous scheduler implies a lower bound of $5f+1$ for the number $n$ of robots for the class of cautious protocols (this bound holds for both ATOM and CORDA models). We also exhibit a deterministic protocol that matches this lower bound (that is, provided that $n\geq 5f+1$, our protocol is determinstic and performs in the CORDA model with fully asynchronous scheduling). Table~\ref{tab:results} summarizes the characteristics of our protocol with respect to previous work on Byzantine tolerant robot convergence (better characteristics for a protocol are depicted in boldface).

\begin{table}
\centering
\begin{tabular}{|c|l|l|l|}\hline
\textbf{Reference} & \textbf{Computation Model} & \textbf{Scheduler}& \textbf{Bounds} \\\hline
\cite{agmon2004ftg} & ATOM & fully synchronous& $n>3f$ \\\hline
\cite{BGT09c} & ATOM & fully synchronous & $n>2f$ \\
& ATOM & $k$-bounded & $n>3f$\\
& \textbf{CORDA} & $k$-bounded & $n>4f$\\\hline
\cite{BGT09d} & \textbf{CORDA} & $k$-bounded & $n>3f$\\\hline
\textbf{This paper} & \textbf{CORDA} & \textbf{fully asynchronous}& $n>5f$\\\hline
\end{tabular}
\caption{Byzantine resilience bounds for deterministic convergence}
\label{tab:results}
\end{table}

\paragraph{Outline}
The remaining of the paper is organized as follows:
Section~\ref{sec:model} presents our model and robot network
assumptions. This section also presents the formal specification of
the convergence problem. Section \ref{sec:lowerbound} presents the byzantine
resilience lower bound
proof. Section~\ref{sec:algorithm} describes our
protocol and its complexity, 
while concluding remarks are presented in Section~\ref{sec:conclusion}.

\section{Model and Problem Definition}
\label{sec:model}

Most of the notions presented in this section are borrowed from\cite{SY99,Pre01,agmon2004ftg}.  We consider a network that consists of a finite set of robots arbitrarily deployed in a uni-dimensional space. The robots are devices with sensing, computing and moving capabilities. They can observe (sense) the positions of other robots in the space and based on these observations, they perform some local computations that can drive them to other locations. 

In the context of this paper, the robots are \emph{anonymous}, in the sense that they can not be distinguished using their appearance, and they do not have any kind of identifiers that can be used during the computation. In addition, there is no direct mean of communication between them. Hence, the only way for robots to acquire information is by observing their positions. Robots have \emph{unlimited visibility}, \emph{i.e.} they are able to sense the entire set of robots. Robots are also equipped with a strong multiplicity sensor that provides robots with the ability to detect the exact number of robots that may simultaneously occupy the same location. We assume that the robots cannot remember any previous observation nor computation performed in any previous step. Such robots are said to be \emph{oblivious} (or \emph{memoryless}). 

A \emph{protocol} is a collection of $n$ \emph{programs}, one operating on each robot. The program of a robot consists in executing {\em Look\mbox{-}Compute\mbox{-}Move cycles} infinitely many times. That is, the robot first observes its environment (Look phase). An observation returns a snapshot of the positions of all robots within the visibility range. In our case, this observation returns a snapshot (also called \emph{configuration} hereafter) of the positions of \emph{all} robots denoted with $P(t)= \{P_1(t), ... , P_n(t)\}$. The positions of correct robots are referred as $U(t)=\{U_1(t), ... , U_m(t)\}$ where $m$ denotes the number of correct robots. Note that $U(t) \subseteq P(t)$. The observed positions are \emph{relative} to the observing robot, that is, they use the coordinate system of the observing robot. We denote by $P^i(t)= \{P_1^i(t), ... , P_n^i(t)\}$ the configuration $P(t)$ given in terms of the coordinate system of robot $i$ ($U^i(t)$ is defined similarly).
Based on its observation, a robot then decides --- according to its program --- to move or to stay idle (Compute phase). When a robot decides a move, it moves to its destination during the Move phase.  An \emph{execution} $e=(c_0, \ldots, c_t, \ldots)$ of the system is an infinite sequence of configurations, where $c_0$ is the initial configuration\footnote{Unless stated otherwise, we make no specific assumption regarding the respective positions of robots in initial configurations.} of the system, and every transition $c_i \rightarrow c_{i+1}$ is associated to the execution of a subset of the previously defined actions. 

A \emph{scheduler} is a predicate on computations, that is, a scheduler defines a set of \emph{admissible} computations, such that every computation in this set satisfies the scheduler predicate. A \emph{scheduler} can be seen as an entity that is external to the system and selects robots for execution. As more power is given to the scheduler for robot scheduling, more different executions are possible and more difficult it becomes to design robot algorithms. In the remaining of the paper, we consider that the scheduler is \emph{fully asynchronous}, that is, in any infinite execution, every robot is activated infinitely often, but there is no bound for the ration between the most activated robot and the least activated one.

We now review the main differences between the ATOM~\cite{SY99} and CORDA~\cite{Pre01} models. In the ATOM model, whenever a robot is activated by the scheduler, it performs a \emph{full} computation cycle. Thus, the execution of the system can be viewed as an infinite sequence of rounds. In a round one or more robots are activated by the scheduler and perform a computation cycle. The \emph{fully-synchronous ATOM} model refers to the fact that the scheduler activates all robots in each round, while the regular \emph{ATOM} model enables the scheduler to activate only a subset of the robots.
In the CORDA model, robots may be interrupted by the scheduler after performing only a portion of a computation cycle. In particular, phases (Look, Compute, Move) of different robots may be interleaved. For example, a robot $a$ may perform a Look phase, then a robot $b$ performs a Look-Compute-Move complete cycle, then $a$ computes and moves based on its previous observation (that does not correspond to the current configuration anymore). As a result, the set of executions that are possible in the CORDA model are a strict superset of those that are possible in the ATOM model. So, an impossibility result that holds in the ATOM model also holds in the CORDA model, while an algorithm that performs in the CORDA model is also correct in the ATOM model. Note that the converse is not necessarily true.
 
The faults we address in this paper are \emph{Byzantine} faults. A byzantine (or malicious) robot may behave in arbitrary and unforeseeable way. In each cycle, the scheduler determines the course of action of faulty robots and the distance to which each non-faulty robot will move in this cycle. However, a robot $i$ is guaranteed to move a distance of at least $\delta_i$ towards its destination before it can be stopped by the scheduler.

Our convergence algorithm performs operations on multisets. A multiset or a bag $S$ is a generalization of a set where an element can have more than one occurrence. The number of occurrences of an element \emph{a} is referred as its \emph{multiplicity}. The total number of elements of a multiset, including their repeated occurrences, is referred as the \emph{cardinality} and is denoted by $|S|$. $\min(S)$(resp. $\max(S)$) is the smallest (resp. largest) element of $S$. If $S$ is nonempty, $range(S)$  denotes the set $[\min(S), \max(S)]$ and $diam(S)$ (diameter of $S$) denotes $\max(S) - \min(S)$.


Given an initial configuration of $n$ autonomous mobile robots ($m$ of which are correct such that $m \geq n-f$), the \emph{point convergence problem} requires that all correct robots asymptotically approach the exact same, but unknown beforehand, location. In other words, for every $\epsilon > 0$, there is a time $t_\epsilon$ from which all correct robots are within distance of at most $\epsilon$ of each other.

\begin{definition}[Byzantine Convergence]
\label{def:byz-convergence}
A system of oblivious robots satisfies the Byzantine convergence specification if and only if $\forall \epsilon > 0, \exists t_\epsilon$ such that $\forall t > t_\epsilon$, $\forall i,j \leq m$, $\mathit{distance}(U_i(t), U_j(t)) < \epsilon$, where $U_i(t)$ and $U_j(t)$ are the positions of some \emph{correct} robots $i$ and $j$ at time $t$, and where $\mathit{distance}(a,b)$ denote the Euclidian distance between two positions. 
\end{definition}

Definition~\ref{def:byz-convergence} requires the convergence property only from the \emph{correct} robots. Note that it is impossible to obtain the convergence for all robots since Byzantine robots may exhibit arbitrary behavior and never join the position of correct robots.

\section{Impossibility for $n \leq 5f$ and a fully asynchronous scheduler}
\label{sec:lowerbound}

In this section we prove the fact that, when the number of robots in the network does not exceed $5f$ (with $f$ of those robots possibly being Byzantine), the problem of Byzantine resilient convergence is impossible to solve under a fully asynchronous scheduler. The result is proved for the weaker ATOM model, and thus extends to the CORDA model.

Our proof is based on a particular initial setting from which we prove that no cautious convergence algorithm is possible if the activation of robots is handled by a fully asynchronous scheduler. Consider a network $N$ of $n$ robots placed on a line segment $[A,B]$, $f$ of which may be Byzantine with $n \leq 5f$. We consider that robots are ordered from left to right. This order is only given for ease of presentation of the proof and is unknown to robots that can not use it in their algorithms. It was proved in \cite{BGT09c} that the problem is impossible to solve when $n\leq 3f$, we thus consider here the case when $3f <n \leq 5f$ only. The initial placement of correct robots is illustrated in Figure~\ref{fig-configuration-c1}: $f$ robots are at location $A$, $f$ others robots are at location $B$ and the remaining $m-2f$ ones are located at some intermediate location between $A$ and $B$. The impossibility proof depends on the ability of the adversary to move these $m-2f$ robots along $[A,B]$, so their position is denoted by a variable $X$, with $X$ belonging to interval $(A, B)$. In the following, these three groups of robots located at $A$, $B$ and $X$ will be referred as $SetA$, $SetB$ and $SetX$ respectively. The positions of the Byzantine robots are determined by the adversary.

\begin{figure}[htbp]
\centering
\includegraphics[height=2cm,width=5cm]{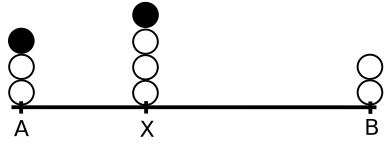}
\caption{Robot Network $N$ (Configuration $C_1$) for $(n=9, f=2)$}
\label{fig-configuration-c1}
\end{figure}

We show by contradiction that in these conditions, no cautious convergence algorithm is possible. Assume that there exists a cautious convergence algorithm $P$ that is correct when the robots are activated by a fully asynchronous scheduler, then we show that in this setting, any cautious algorithm $P$ satisfies properties that can by used by the adversary to prevent convergence of $P$, which is a contradiction.

The properties satisfied by all cautious protocols are captured in the following two basic facts:

\begin{itemize}
\item \textbf{Fact 1:} If all \textit{Byzantine} robots are inside $[A, X]$ (resp. $[X, B]$) then when robots of $SetA$ (resp. $SetB$) are activated, their calculated destination points are necessarily inside $[A, X]$ (resp. $[X, B]$). This fact is proved by Lemma \ref{lemma_fact1}.

\item \textbf{Fact 2:} The adversary is able to move the robots of $SetX$ as close as desired to location $A$ (resp. $B$). This is proved by Lemmas~\ref{lemma_fact2_even}, \ref{lemma_fact2_odd} and \ref{lemma-fact2}.
\end{itemize}

Based on this, the adversary first moves the robots of $SetX$ very close to $A$ (using Fact 2) and then activates the robots of $SetA$ that remain in the neighborhood of $A$ (due to Fact 1). Afterward, it moves the intermediate robots of $SetX$ very close to $B$ (using Fact 2) and activates the robots of $SetB$ which also remain in the neighborhood of $B$ (due to Fact 1). By repeating these actions indefinitely, the adversary ensures that every robot is activated infinitely often in the execution yet prevents convergence at the same time since robots at $A$ and $B$ remain always arbitrarily close to their initial positions and never converge.

\medskip

In the following, we prove $Fact 1$ and $Fact 2$ by a sequence of lemmas, and then give a formal presentation of the algorithm used by the adversary to prevent any cautious protocol from achieving convergence.

\begin{lem}
\label{lemma_fact1}
In $N$, $\forall X \in (A, B)$, 
if all Byzantine robots are inside $[A, X]$ (resp. $[X, B]$) then when robots of $SetA$ (resp. $SetB$) are activated, 
their destination points computed by any cautious algorithm are necessarily inside $[A, X]$ (resp. $[X, B]$).
\end{lem}

To prove $Fact 2$, we use the network $N$ described above (see Figure~\ref{fig-configuration-c1}). We prove only the capability of the adversary to move the intermediate robots at $X$ as close as wanted to $B$, the other case being symmetric. $Fact 2$ implies that if the number of robots in the network is lower or equal to $5f$ then it always exists a judicious placement of the Byzantine robots that permits the adversary to make the intermediate robots in $X$ move in the direction of $B$ up to a location that is as close as desired to $B$. We divide the analysis in two cases depending on the parity of $(n-f)$.

\paragraph{Case 1: $(n-f)$ is even}

To push the robots of $SetX$ as close to $A$ or $B$ as wanted, the adversary uses algorithm $GoToBorder1$ $(G2B1)$ described as Algorithm~\ref{alg-gotoborder-even}. Informally, the algorithm divides Byzantines robots between position $X$ and the target border to which the adversary wants to push the robots of $SetX$ (\emph{e.g.} $B$ in what follows). The aim of the adversary is to maintain the same number of robots in $X$ and $B$ (this is possible because $n-f$ is even). We prove that in this case, any cautious convergence algorithm makes the robots of $SetX$ move towards $B$. However, the distance traveled by them may be too small to bring them sufficiently close to $B$. Since the scheduler is fully asynchronous, it is authorized to activate the robots of $SetX$ as often as necessary to bring them close to $B$, as long as it does so for a finite number of times.

\begin{algorithm}
\caption{GoToBorder1 (G2B1)}          
\begin{algorithmic}
\small
\STATE   \textbf{Input:} $Border$: the border towards which robots of $SetX$ move (equal to $A$ or $B$).
\STATE   \textbf{Input:} $d$: a distance.\\
\STATE
\STATE   \textbf{Actions:}\\
\WHILE {$distance(X,Border)>d$}
\STATE Place $(n-3f)/2$ byzantine robots at $Border$.
\STATE Place $(5f-n)/2$ byzantine robots at $X$.
\STATE Activate simultaneously all robots of $SetX$ and make them move to their computed destination $D$.
\STATE $X \leftarrow D$
\ENDWHILE
\end{algorithmic}
\label{alg-gotoborder-even}                  
\end{algorithm}

\begin{lem}
\label{lemma_fact2_even}
If $(n-f)$ is even, $\forall d<distance(A,B)$, $\forall Border \in {A,B}$, algorithm $G2B1 (Border, d)$ terminates.
\end{lem}

\paragraph{Case 2: $(n-f)$ is odd}

To prove Lemma~\ref{lemma_fact2_even}, we relied on the symmetry induced by the placement of Byzantine robots. This symmetry is possible only because $(n-f)$ is even. Indeed, having the same number of robots in $B$ and $X$ implies that convergence responsibility is delegated to both robots at $X$ and at $B$ (there is no asymmetry to exploit to get one of these two groups play a role that would be different from the other group. Robots of $SetX$ and $SetB$ have thus no other choice but to move toward each other when they are activated. The distance traveled at each activation must be large enough to ensure the eventual convergence of the algorithm. 

However, the situation is quite different when $(n-f)$ is odd. Indeed, the number of robots is necessarily different in $X$ and $B$, which means that one of the two points has a greater multiplicity than the other. Then in this case there is no guarantee that a cautious convergence algorithm will order the robots of $SetX$ to move toward $B$ when they are activated (the protocol could delegate the convergence responsibility to robots of $SetB$). Nevertheless, we observe that whatever the cautious algorithm is, if it does not move the robots that are located at the greatest point of multiplicity, it must do so for those at the smallest one (and \emph{vice versa}), otherwise no convergence is ever possible. The convergence is thus either under the responsibility of robots at the larger point of multiplicity or those at the smaller one (or both). 

This observation is exploited by Algorithm $GoToBorder2$ $(G2B2)$ that is presented as Algorithm~\ref{alg-gotoborder-odd}, that tries the two possible cases to ensure its proper functioning when confronted to any cautious algorithm. The algorithm forms the larger point of multiplicity at $B$ at one cycle, and the next cycle at $X$. Thus, point $X$ will be the larger point of multiplicity one time, and the smallest one the next time. This implies that the robots of $SetX$ must move towards $B$ at least once every two cycles. So by repeatedly alternating between the two configurations where robots of $SetX$ are successively the set of larger and smaller multiplicity, the adversary ensures that they end up moving towards $B$. The fully asynchrony of the scheduler ensures that they are activated as many times as it takes to move them as close to $B$ as wanted, provided that the algorithm terminates.

\begin{algorithm}
\begin{algorithmic}
\small
\STATE   \textbf{Input:} $Border$: the border towards which the robots of $SetX$ move (equal to $A$ or $B$).
\STATE   \textbf{Input:} $d$: a distance.\\
\STATE

\STATE   \textbf{Actions:}\\

\STATE Place $(n-3f+1)/2$ Byzantine robots at $Border$.
\STATE Place $(5f-n-1)/2$ Byzantine robots at $X$.
\WHILE {$distance(X,Border)>d$}
\STATE Activate simultaneously all robots at $X$ and make them move to their computed destination $D$.
\STATE $X \leftarrow D$.
\STATE Move a Byzantine robot from $Border$ to $X$.
\STATE Activate simultaneously all robots at $X$ and make them move to their computed destination $D$.
\STATE $X \leftarrow D$
\STATE Move a Byzantine robot from $X$ to $Border$.
\ENDWHILE
\end{algorithmic}
\caption{GoToBorder2 (G2B2)}          
\label{alg-gotoborder-odd}                  
\end{algorithm}

\begin{lem}
\label{lemma_fact2_odd}
If $(n-f)$ is odd, $\forall d<distance(A,B)$, $\forall Border \in {A,B}$, algorithm $G2B2 (Border, d)$ terminates.
\end{lem}

We are now ready to prove $Fact 2$.

\begin{lem}
\label{lemma-fact2}
For $n \leq 5f$, $\forall d<distance(A,B)$, if the robots run a cautious convergence algorithm, the fully distributed scheduler is able to move the robots of $SetX$ into a position $\geq B-d$ or $\leq A+d$.
\end{lem}

\begin{proof}
The proof follows directly from Lemmas~\ref{lemma_fact2_even} and \ref{lemma_fact2_odd}.
\end{proof}

\paragraph{The Split function}

The purpose of Algorithms $G2B1$ and $G2B2$ is to push the intermediate robots of $SetX$ as close the adversary want to the extremities of the network. For ease of the description, we assume in what follows that the adversary want to push them towards the extremity $B$. These two routines are then used by the adversary to prevent the convergence of the algorithm. For the algorithm of the adversary to work, it is necessary to keep the robots of $SetA$, $SetB$ and $SetX$ separated from each other and to avoid for example that the robots of $SetX$ merge with those of $SetB$ and form a single point of multiplicity. Yet, functions $G2B1$ and $G2B2$ cannot prevent such a situation to appear because the destinations are computed by the convergence algorithm which can order the robots to move exactly towards $B$. If the distance to travel is too small ($distance (X, B) \leq \delta_i$ for all $i \in SetX$), then the adversary can not stop the robots of $SetX$ before they arrive at $B$.
To recover from this situation and separate the robots that have merged, we define a new function $Split (Set, Border)$ which separates the robots of $Set$ from those located at $Border$. For example, $Split (SetX, B)$ separates the robots of $SetX$ from those of $SetB$ by directing them towards $A$. Lemma~\ref{lem:pour_split} is used to prove that function $Split$ performs as planned.
Let $N$ be a network of $n$ robots divided between two positions $A$ and $B$. let $p$ and $q$ be the number of robots at $A$ and $B$ respectively. These robots are endowed with a cautious convergence algorithm that tolerate the presence of up to $f$ Byzantine robots. Lemma~\ref{lem:pour_split} proves that if a robot in $A$ or $B$ is activated, it cannot remain in its position and moves toward the robots located in the other point.

\begin{lem}
\label{lem:pour_split}
If $|p-q| \geq f$, then if a robot at $A$ (resp. $B$) is activated, its destination computed by any cautious convergence algorithm lays inside $(A, B]$ (resp. $[A, B)$).
\end{lem}

We now present Function $Split(Set, Border)$ that is presented as Algorithm~\ref{function split}. 
We first define $Max{\delta}$ as $max\{\delta_i / i$ is a correct robot$\}$ such that $\delta_i$ is the minimum distance that can be traveled by a robot $i$ before it may be stopped by the adversary. This means that if a group of robots ($SetX$ in our case) are distant from their destination by more than $Max{\delta}$, the adversary is able to stop them all before they reach their destination. Notice now that in the setting of network $N$ described in Figure~\ref{fig-configuration-c1}, $SetA$ and $SetB$ contain each exactly $f$ correct robots.
If robots of $SetX$ merge with those of $SetB$ for example, they form a set of $n-2f$ correct robots colocated in the same multiplicity point. By placing all the Byzantine robots at $A$, this location contains a set of $2f$ robots. The difference between the two sets of robots in $A$ and $B$ is lower or equal to $f$ (because $3f < n \leq 5f$). Then if we activate the robots of $SetX$ (which are located at $B$), they will move towards $A$ according to Lemma~\ref{lem:pour_split}. By stopping these robots once they all travelled a distance equal to $Max{\delta}$ or reached they destination before, we ensure that the three sets $SetA$, $SetX$ and $SetB$ are disjoint, because the initial distance between A and B is $>Max{\delta}$. 

\begin{algorithm}
\begin{algorithmic}
\small
\REQUIRE $distance(A,B) > Max{\delta}$
\STATE
\STATE \textbf{Variables:}
\STATE \textbf{Input:} $Border$: is equal to $A$ or to $B$.
\STATE \textbf{Input:} $Set$: the set of robots to move away from $Border$.
\STATE $OppositeBorder$: is equalt to $B$ if the input $Border$ is equal to $A$, and vice versa.
\STATE
\STATE \textbf{Actions:}\\
\STATE Place all Byzantine robots in $OppositeBorder$.
\STATE Activate the robots of $Set$, and stop them at a point $Max{\delta}$ away from $Border$.
\end{algorithmic}
\caption{Function Split(Set, Border)}          
\label{function split}                  
\end{algorithm}

\paragraph{The fully asynchronous scheduler algorithm}

\begin{thm}
In the ATOM model, the problem of Byzantine resilient convergence is impossible to solve with a cautious algorithm under a fully asynchronous scheduler.
\end{thm}

\begin{algorithm}
\begin{algorithmic}
\small
\REQUIRE $distance(A,B) > Max{\delta}$
\STATE

\STATE \textbf{Definitions}:\\

\STATE $d_0$: any distance that is strictly smaller than $distance(A,B)/4$, let $d_0 \leftarrow distance(A,B)/10$.
\STATE $G2B(Border, d)$: equal to $G2B1(Border, d)$ if $n-f$ is even and equal to $G2B2(Border, d)$ if $n-f$ is odd
\STATE
\STATE \textbf{Actions:}\\

\WHILE {true}

\STATE $G2B(A, d_0)$.
\STATE Activate the robots at $A$.
\STATE if the robots of SetX are at $A$, then $Split(SetX, A)$.
\STATE $G2B(B, d_0)$.
\STATE Activate the robots at $B$.
\STATE if the robots of $SetX$ are at $B$, then $Split(SetX, B)$.
\STATE $d_0 \leftarrow d_0/2$

\ENDWHILE
\end{algorithmic}
\caption{Adversary Algorithm}          
\label{alg_adversary}                  
\end{algorithm}

\begin{proof}
We prove that for network $N$, there can be no cautious convergence algorithm for $n \leq 5f$ if the robots are activated by a fully asynchronous scheduler. The algorithm of the adversary is given as Algorithm~\ref{alg_adversary} and it can prevent any cautious algorithm to converge. Indeed, if the initial distance between robots at $A$ and $B$ is equal to $d$, then these robots will always remain distant from each other by a distance at least equal to $6d/10$. The proof of algorithm $\ref{alg_adversary}$ follows directly from Lemmas~\ref{lemma_fact1}, \ref{lemma-fact2} and \ref{lem:pour_split}. 
\end{proof}

\section{Deterministic Asynchronous Convergence}
\label{sec:algorithm}


In this section, we propose a deterministic convergence algorithm and
prove its correctness in CORDA model 
under a fully asynchronous scheduler when there are at least $5f+1$ robots, $f$ of which may be Byzantine.

\paragraph{Algorithm Description}
The idea of our algorithm is based on three mechanisms: (1) a trimming
function for the computation of destinations, 
(2) location dependency and (3) an election procedure.
The purpose of the trimming function is to ignore the most extreme
positions in the network when computing the destination. 
Robots move hence towards the center of the remaining positions.
Consequently, the effect of Byzantine robots is canceled since they 
cannot drag the correct robots away from the range of correct positions.

Location dependency affects the computation of the trimming function such that the returned result depends on the position of the calling robot. This leads to interesting properties on the relation between the position of a robot and its destination that are critical to convergence.
The election procedure instructs to move only the robots located at the two extremes of the network.
Thus, by the combined effect of these three mechanisms, as the
algorithm progresses, the extreme 
robots come together towards the middle of the range of correct positions which ensures the eventual convergence of the algorithm.

The algorithm uses three functions as follows. 
The trimming function $trim_{2f}^i()$ removes among the $2f$ largest positions of the multiset given in parameter \emph{only} those that are greater than the position of the calling robot $i$. Similarly, it removes among the $2f$ smallest positions only those that are smaller than the position of the calling robot. It is clear that the output of $trim_{2f}^i()$ depends on the position of the calling robot.
Formally, let $minindex_i$ be the index of the minimum position between $P_i(t)$ and $P_{2f+1}(t)$ (if $P_i(t)<P_{2f+1}(t)$ then $minindex_i$ is equal to $i$, otherwise it is equal to $2f+1$). Similarly, let $maxindex_i$ be the index of the maximum position between $P_i(t)$ and $P_{n-2f}(t)$ (if $P_i(t)> P_{n-2f}(t)$ then $maxindex_i$ is equal to $i$, otherwise it is equal to $n-2f$). $trim_{2f}^i(P(t))$ is the multiset consisting of positions $\{P_{minindex_i}(t), P_{minindex_i+1}(t), \ldots, P_{maxindex_i}(t)\}$.

The function $center()$ simply returns the median point of the input range. The two functions are illustrated in Figure \ref{fig-illustration-algo}).

The election function returns true if the calling robot is allowed to move. Only the robots that are located at the extremes of the networks are allowed to move, that is those whose position is either $\leq P_{f+1}(t)$ or $\geq P_{n-f}(t)$.

\begin{figure}[htbp]
\begin{center}
\includegraphics[scale=.4]{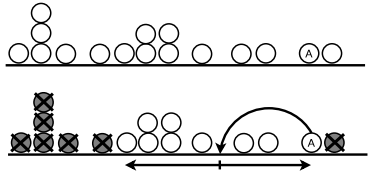}
\caption{Illustration of functions $trim_{2f}^i$ and $center$ for robots $A$ in a system of $(n=16, f=3)$ robots.}
\label{fig-illustration-algo}
\end{center}
\end{figure}

\begin{algorithm}
\begin{algorithmic}
\small
\STATE \textbf{Functions}:\\

\STATE $trim_{2f}^i(P(t))$: 
removes up to $2f$ largest positions that are larger than $P_i(t)$ and up to $2f$ smallest positions that are smaller than $P_i(t)$ from the multiset $P(t)$ given in parameter.
\STATE $center()$: returns the point that is in the middle of the range of points given in parameter.
\STATE $elected() \equiv ((P_i(t) \leq P_{f+1}(t))$ or $(P_i(t) \geq
P_{n-f}(t)))$. This function returns true if the calling robot is allowed to move.
\STATE
\STATE   \textbf{Actions}:\\

\STATE if $elected()$ move towards $center(trim_{2f}^i(P(t)))$
\end{algorithmic}
\caption{Convergence Algorithm under a fully asynchronous Scheduler}          
\label{alg-convergence-5f}                  
\end{algorithm}


By definition, convergence aims at asymptotically decreasing the range of possible positions for the correct robots. The shrinking property captures this property. An algorithm is shrinking if there exists a constant factor $\alpha \in (0,1)$ such that starting in any configuration the range of correct robots eventually decreases by a multiplicative $\alpha$ factor. Note that to deal with the asynchrony of the model, the diameter calculation takes into account both the positions and destinations of correct robots.

\begin{definition}[Shrinking Algorithm] 
\label{def:shrink}
An algorithm is \emph{shrinking} if and only if $\exists \alpha \in (0,1)$ such that $\forall t, \exists t^\prime > t$, such that $diam(U(t^\prime) \cup D(t^\prime)) < \alpha*diam(U(t) \cup D(t))$, where $U(t)$ and $D(t)$ are respectively the the multisets of positions and destinations of correct robots.
\end{definition}

A natural way to solve convergence is to never let the algorithm increase the diameter of correct robot positions. In this case the algorithm is called \emph{cautious}. This notion was first introduced in~\cite{dolev1986raa}. A cautious algorithm is particularly appealing in the context of Byzantine failures since it always instructs a correct robot to move inside the range of the positions held by the correct robots regardless of the locations of Byzantine ones. The following definition introduced first in \cite{BGT09c} customizes the definition of cautious algorithm proposed in \cite{dolev1986raa} to robot networks.

\begin{definition}[Cautious Algorithm]
\label{def:caut}
Let $D_i(t)$ be the last destination calculated by the robot $i$ before time $t$ and let $U^i(t)$ the positions of the correct robots as seen by robot $i$ before time $t$. \footnote{If the last calculation was executed at time $t^\prime \leq t$ then $D_i(t) = D_i(t^\prime)$.} An algorithm is \emph{cautious} if it meets the following conditions:
\emph{(i)} \textbf{cautiousness: } $\forall t,~D_i(t) \in range(U^i(t))$ for each robot $i$, and \emph{(ii)} \textbf{non-triviality: } $\forall t$, if $diameter(U(t)) \neq 0$ then $\exists t^\prime>t$ and a robot $i$ such that $D_i(t^\prime)\neq U_i(t^\prime)$ (at least one correct robot changes its position whenever convergence is not achieved).
\end{definition}


\begin{thm}\cite{BGT09c}
\label{th:cands}
Any algorithm that is both cautious and shrinking solves the convergence problem in faulty robots networks.
\end{thm}

In the appendix we prove the correctness of Algorithm
\ref{alg-convergence-5f} in the CORDA 
model under a fully asynchronous scheduler. In order to show that Algorithm
\ref{alg-convergence-5f} converges, 
we prove first that it is cautious then we prove that it satisfies the
specification of a shrinking algorithm. 
Convergence then follows from Theorem \ref{th:cands}.


%

\section{Concluding remarks}
\label{sec:conclusion}

Our work closes the study of the convergence problem for 
unidimensional robot networks.
We studied the convergence problem under the most generic settings: 
asynchronous robots under unbounded adversaries and byzantine fault model.
We proved that in these settings the byzantine resilience lower bound is $5f+1$ and 
we propose and prove correct the first deterministic convergence algorithm that meets this lower bound.
We curently investigate the extension of the curent work to 
the multi-dimensional spaces.

\begin{footnotesize}
\bibliographystyle{plain}
\bibliography{convergence}
\end{footnotesize}

\newpage

\pagenumbering{roman}

\section*{Appendix}

\subsection*{Proof of Lemma~\ref{lemma_fact1}}

\begin{proof}
We prove the lemma only for the case when all Byzantine robots are inside $[A, X]$, and we denote the corresponding configuration by $C_1$ (see Figure~\ref{fig-configuration-c1}). The case where all Byzantine robots are inside $[X, B]$ is symmetric.

\begin{figure}[htbp]
\centering
\includegraphics[height=2cm,width=5cm]{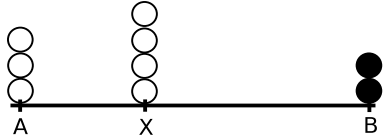}
\caption{Illustration of Lemma $\ref{lemma_fact1}$, configuration $C_2$}
\label{fig-configuration-c2}
\end{figure}

Let $C_2$ (see Figure~\ref{fig-configuration-c2}) be a similar configuration of $n$ robots where the distribution of positions is isomorphic to that of $C_1$, but where the correct and Byzantine robots are located differently: all robots at $B$ are byzantine (there are $f$ such robots), and all robots inside $[A,X]$ are correct. Since the robot convergence algorithm is cautious, the diameter of correct robots in $C_2$ must never decrease, and then all their calculated destination points must lay inside $[A,X]$.
Since $C_1$ and $C_2$ are indistinguishable to individual robots of $SetA$, the Look and Compute phases give the same result in the two cases, which proves our lemma.
\end{proof}

\subsection*{Proof of Lemma~\ref{lemma_fact2_even}}

\begin{proof}
We prove the Lemma by contradiction. We assume that the algorithm does not terminate for a given input distance $d_0$, and we prove that this leads to a contradiction. We consider only the case where $Border=B$, the other case being symmetric. The non-termination of the algorithm implies that there exists some distance $d_1 \leq d_0$ such that robots at $X$ and $B$ always remain distant by at least $d_1$ from each other, even if robots at $X$ are activated indefinitely.

Note that the placement of Byzantine robots in $G2B1$ implies that initially, and for $n\leq 5f$, the number of robots located at $X$ and $B$ is the same and is equal to $(n-f) / 2$ as illustrated in Figure~\ref{fig-fact2-even-configuration-c1}.$(a)$. We denote by $C_1$ the resulting configuration. We now construct a configuration $C_2$ (see Figure~\ref{fig-fact2-even-configuration-c1}.$(b)$) that is isomorphic to $C_1$ but with a different distribution of Byzantine and correct robots: correct robots are divided equally between $X$ and $B$, $(n-f)/2$ correct robots at $X$ and $(n-f)/2$ others at $B$. By hypothesis, these robots are supposed to converge to a single point (located between $X$ and $B$ as the convergence point is computed by a cautious algorithm).

\begin{figure}[htbp]
\centering
\subfigure[{Configuration $C_1$ for $(n=13, f=3)$}]{\includegraphics[scale=.4]{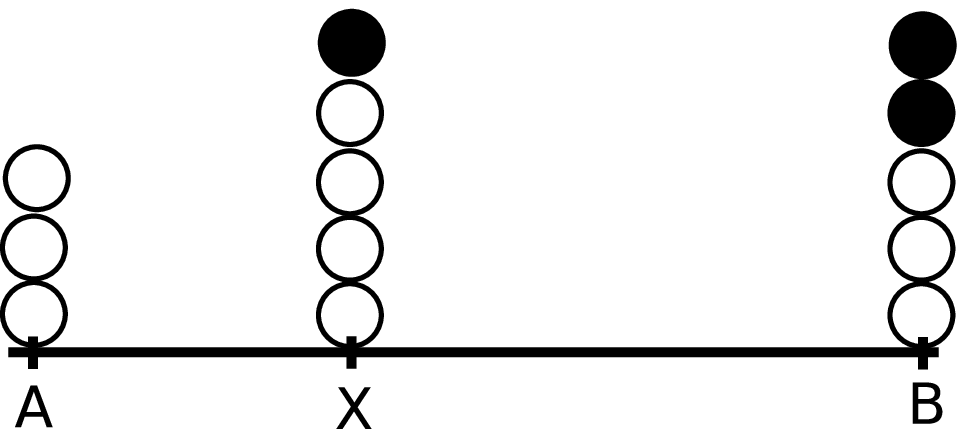}}
\hspace{2cm}
\subfigure[{Configuration $C_2$ for $(n=13, f=3)$}]{\includegraphics[scale=.4]{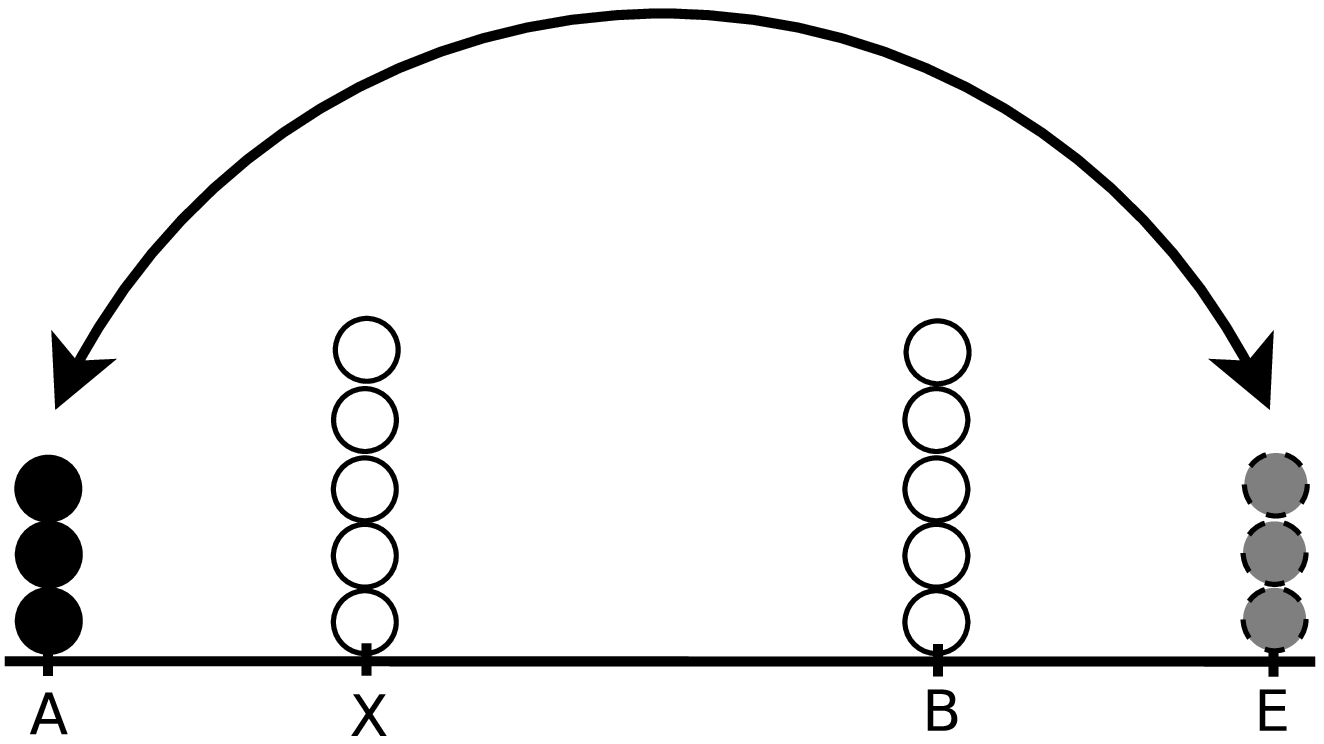}}
\caption{Illustration of lemma $\ref{lemma_fact2_even}$ (Fact2, $(n-f)$ even)}
\label{fig-fact2-even-configuration-c1}
\end{figure}

The placement of Byzantine robots and the choice of activated robots at each cycle is divided into two parts. During even cycles, Byzantine robots are placed at point $A$ and robots located at $X$ are activated. During odd cycles, the scheduler constructs a strictly symmetrical configuration by moving Byzantine robots from $A$ to a point $E$ with $E> B$ and $distance (B, E)=distance (A, X)$. In this case, the scheduler activates robots at $B$.

In these conditions, activating robots at $X$ ensures that they always remain at a distance of at least $d_1$ from those located at $B$ (as in configuration $C_1$). Indeed, configurations $C_1$ and $C_2$ are equivalent and completely indistinguishable to individual robots which must behave similarily in both cases (as the algorithm is deterministic). And by symmetry, the activation of robots at $B$ during odd cycles also ensures that minimum distance of $d_1$ between the two groups of robots. Hence, robots at $X$ and $B$ remain separated by a distance of at least $d_1$ forever even if activated indefinitely, which prevents the convergence of the algorithm and leads to a contradiction. This proves our Lemma.
\end{proof}

\subsection*{Proof of Lemma~\ref{lemma_fact2_odd}}

\begin{proof}
We consider in our proof only the case when $Border=B$ since the other case is symmetric. The placement of Byzantine robots in $G2B2$ is such that the multiplicity of $X$ exceeds that of $B$ by $1$ during even cycles, and lowers it by $1$ during odd cycles. We denote by $C_0$ the initial configuration (in which the multiplicity of $X$ is less than of $B$ by $1$ as illustrated in Figure \ref{fig-fact2-odd-configuration-c0}.$(a)$). 

We assume for the purpose of contradiction that $G2B2$ does not terminate for some input distance $d_0$. This means that robots of $SetX$ and $SetB$ remain always distant from each others by a distance at least equal to $d_1$ with $d_1$ being some distance $\leq d_0$. The resulting execution in this case is denoted by $E_0 = \{C_0, C_1, C_2, C_3, ...\}$. A configuration $C_{i+1}$ is obtained from $C_i$ by activating robots at $X$, letting them execute their Move phases, and moving one Byzantine robot from $X$ to $B$ or \emph{vice versa}.

\begin{figure}[htbp]
\centering
\subfigure[{Initial configuration $C_0$ for $(n=12, f=3)$}]{\includegraphics[scale=.4]{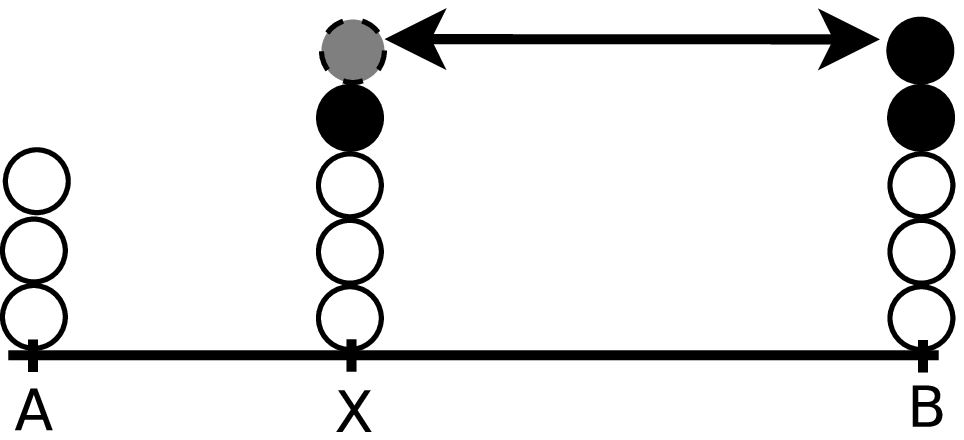}}
\hspace{2cm}
\subfigure[{Initial configuration $C^\prime_0$ for $(n=12, f=3)$}]{\includegraphics[scale=.4]{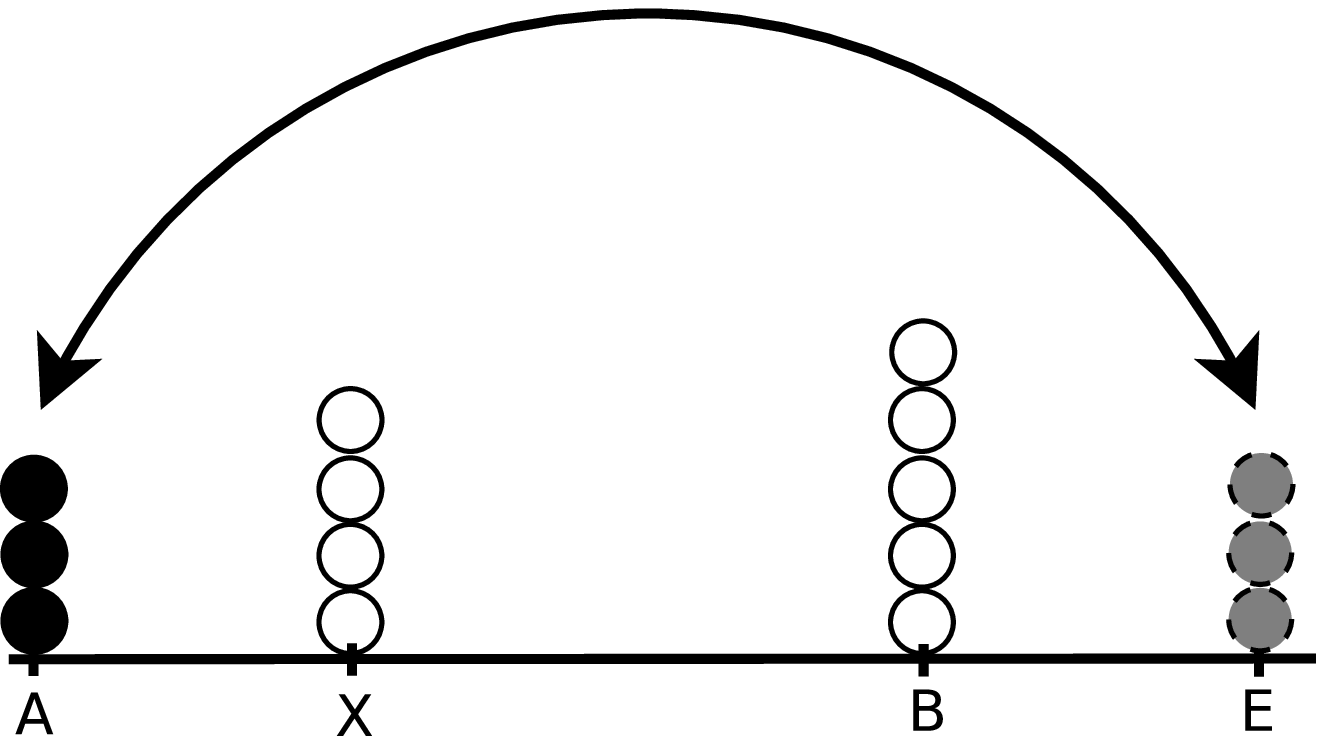}}
\caption{Illustration of lemma $\ref{lemma_fact2_odd}$ (Fact2, $(n-f)$ odd)}
\label{fig-fact2-odd-configuration-c0}
\end{figure}

We construct a configuration $C^\prime_0$ equivalent to $C_0$ but where correct robots are divided between $X$ and $B$ with$ \lfloor (n-f) / 2 \rfloor$ robots at $X$ and $\lceil (n-f)/2 \rceil$ robots at $B$ (see Figure~\ref{fig-fact2-odd-configuration-c0}.$(b)$) . By definition, these robots must converge to a point between $X$ and $B$ since they are endowed with a cautious convergence algorithm. Byzantine robots are at $A$. Since $C^\prime_0$ and $C_0$ are equivalent, the activation of robots at $X$ and the displacement of Byzantine robots to the right of $B$ will produce a configuration $C^\prime_1$ that is equivalent to $C_1$ by symmetry. 

This time, activated robots are those at $B$. By moving them to their calculated destination points and by moving Byzantine robots again to the left of $X$ the scheduler can form a configuration $C^\prime_2$ which is equivalent to $C_2$.

This process can be repeated: during odd cycles, Byzantine robots are at the left of $X$ and robots at $X$ are activated. During even cycles, the situation is symmetrical: Byzantine robots are to the right of $B$ and robots at $B$ are activated. The obtained execution $E^\prime_0 = \{C^\prime_0, C^\prime_1, C^\prime_2, C^\prime_3, ...\}$ is equivalent to $E_0$, and robots at $X$ and $B$ remain separated by a distance at least equal to $d_1$ forever even if they are activated indefinitely. This prevents the convergence of the convergence protocol while ensuring fairness of activations, which contradicts the assumptions and proves our Lemma.
\end{proof}

\subsection*{Proof of Lemma~\ref{lem:pour_split}}

\begin{proof}
Let $C_1$ be the initial configuration, and consider the computed destination by an activated robot located at $A$ (the case of a robot located at $B$ is symmetric). Since the algorithm is cautious, this destination point is necessarily located inside $[A, B]$. For the lemma to be correct, it suffices to prove that this destination is different from $A$. In other words, we must prove that the robot moves towards $B$ upon its activation. So assume for the sake of contradiction that it is not the case, that is, the computed destination is $A$ and let us separate the analysis into three cases depending on the relationship between $p$ and $q$:

\begin{itemize}

\item \textbf{Case 1 ($p>q$):} Let $C_2$ be a configuration isomorphic to $C_1$ with the following placement of robots: At $A$ there are $min(f, p)$ Byzantine robots and $p-min(f,p)$ corrects ones, and at $B$ are located $f-min(f,p)$ Byzantine robots and $q-f+min(f,p)$ corrects ones. Since Configurations $C_1$ and $C_2$ are indistinguishable to individual robots, the destinations computed in the two cases are the same. So when the robots at $A$ are activated, they do not move. The next cycle, the adversary moves $p-q$ Byzantine robots from $A$ to $B$ to obtain a configuration $C_3$ symmetric to $C_2$. This time, the adversary activates the robots at $B$ which do not move either since $C_2$ and $C_3$ are symmetric. Then, the adversary brings the $p-q$ Byzantine robots to $A$ to get again the configuration $C_2$ and then activates the robots at $A$. The process repeats, and by placing these $p-q$ Byzantine robots in one cycle at $A$ and the next cycle at $B$, the adversary prevents the convergence of the algorithm. This is a contradiction.

\item \textbf{Case 2 ($p<q$):} We can reach a contradiction by using an argument similar to Case 1.

\item \textbf{Case 3 ($p=q$):} If the activated robots at $A$ do not move upon their activation, it is also the case at $B$ since the configuration is symmetric. This prevents the convergence of the algorithm and leads also to a contradiction.
\end{itemize}

Consequently, the lemma is proved.
\end{proof}

\subsection*{Proof of Algorithm~\ref{alg-convergence-5f}}

\subsubsection*{Algorithm \ref{alg-convergence-5f} is cautious}

In this section we prove that Algorithm \ref{alg-convergence-5f}  is a cautious algorithm (see Definition \ref{def:caut}) for $n>5f$. The following lemma states that the range of the trimmed multiset $trim_{2f}^i(P(t))$ is contained in the range of correct positions.


\begin{lem}
\label{trim_sub_range}
Let $i$ be a correct robot executing Algorithm \ref{alg-convergence-5f}, 
it holds that 
\[\forall t, range(trim_{2f}^i(P(t))) \subseteq range(U(t))
\]
\end{lem}

\begin{proof}
We prove that for any correct robot, $i$, the following conditions hold: 
\begin{enumerate}
\item $\forall t,~ min(trim_{2f}^i(P(t))) \in range(U(t)).$
\item $\forall t,~ max(trim_{2f}^i(P(t))) \in range(U(t)).$
\end{enumerate}

\begin{enumerate}
\item 
By definition, $min(trim_{2f}^i(P(t)))=min\{P_i(t), P_{2f+1}(t)\}$. 
Hence proving Property (1) reduces to proving 
$P_i(t) \in range(U(t))$ and $P_{2f+1}(t) \in range(U(t))$.

\begin{enumerate}
\item $P_i(t) \in range(U(t))$ directly follows from the assumption that robot $i$ is correct.
\item $P_{f+1}(t) \in range(U(t))$.
Suppose the contrary: there exists some time instant $t$ such that $P_{2f+1}(t) \notin range(U(t))$ 
and prove that this leads to a contradiction. If $P_{2f+1}(t) \notin range(U(t))$ 
then either $P_{2f+1}(t) < U_1(t)$ or $P_{2f+1}(t) > U_m(t)$.

\begin{enumerate}
\item If $P_{2f+1}(t) < U_1(t)$ then there are at least $2f+1$ positions $\{P_1(t),$ $P_2(t),$ $\ldots,$ $P_{2f}(t),$ $P_{2f+1}(t)\}$ 
that are smaller than $U_1(t)$ which is the \emph{first} correct position in the network at time $t$. This means that 
there would be at least $2f+1$ byzantine robots in the system. But this contradicts the assumption 
that at most $f$ byzantine robots are present in the system.

\item If $P_{2f+1}(t) > U_m(t)$ then since $n>5f$ there are more than $3f$ positions $\{P_{2f+1}(t), ...,P_n(t)\}$ 
that are greater than $U_m(t)$, which is the \emph{last} correct position in the system at time $t$. This also leads to a contradiction.\\

\end{enumerate}
\end{enumerate}

\item The property is symmetric to 2) and can be proved using the same argument.
\end{enumerate}
\end{proof}


A direct consequence of the above property is that correct robots always compute a destination within the range of positions held by correct robots, whatever the behavior of Byzantine ones. Thus, the diameter of positions held by correct robots never increases. Consequently, the algorithm is cautious. The formal proof is proposed in the following lemma.

\begin{lem}
\label{cautious-FS}
Algorithm \ref{alg-convergence-5f} is cautious for $n>5f$.
\end{lem}

\begin{proof}
We have to prove the two properties of cautious algorithms, namely cautiousness and non-triviality.

\textbf{Cautiouness: }We start by the cautiousness property of our algorithm. 
According to Lemma \ref{trim_sub_range}, $range(trim_{2f}^i(P(t))) \subseteq
range(U(t))$ for each correct robot $i$, 
thus $center(trim_{2f}^i(P(t))) \in range(U(t))$. 
It follows that all destinations computed by correct robots are located 
inside $range(U(t))$ which proves the cautiousness property.

\textbf{Non-triviality: }By fairness, the robots at positions $U_1(t)$ and $U_m(t)$ are guaranteed to be eventually elected irrespective of the positions of byzantine robots. And at least one of them will move unless all correct robots are colocated in the same point.This proves the non-triviality condition.

\end{proof}


\subsubsection{Algorithm \ref{alg-convergence-5f} is Shrinking}

The following lemma proves that the only robots that can be elected
are those located at the extremes of the network, namely those whose
position is either less equal than $U_{f+1}(t)$ or greater equal than $U_{m-f}(t)$. The activation of these robots move them away from the extremes of the network, thereby reducing the diameter of positions held by correct robots which leads to convergence.

\begin{lem}
\label{lem:i_inf_f+1}
If some correct robot $i$ is activated at time $t$, then either
$U_i(t) \leq U_{f+1}(t)$ or $U_i(t) \geq U_{m-f}(t)$ where $m$ is the
number of correct robots in the network and $U_i(t)$ denotes the
position of correct robot $i$ at $t$; 
\end{lem}

\begin{proof}
By definition of the algorithm, a robot is activated only if its position is either $ \leq P_{f+1}(t)$ or $ \geq P_{n-f}(t)$. To prove the lemma, it suffices then to show that $P_{f+1}(t) \leq U_{f+1}(t)$ and $P_{n-f}(t) \geq U_{m-f}(t)$:

To prove that $P_{f+1}(t) \leq U_{f+1}(t)$, we suppose to the contrary that $P_{f+1}(t) > U_{f+1}(t)$. In this case, $P_{f+1}(t)$ would be strictly greater than all the positions $\{U_1(t), ..., U_{f+1}(t)\}$, which contradicts the definition of $P_{f+1}(t)$ as the $(f+1)$-th position in the network.
This proves that $P_{f+1}(t) \leq U_{f+1}(t)$ and the same argument is used to prove that $P_{n-f}(t) \geq U_{m-f}(t)$, since the two cases are symmetric. 

\end{proof}

The following lemma proves an important property on the relationship
between the position of a robot and its computed destination. Indeed,
knowing the position $U_i(t)$ held by a correct robot $i$ at time $t$,
it is possible to give bounds on the possible 
value of its destination point $D_i(t)$. Interestingly, this bound holds irrespective of the positions of Byzantine robots and the actions of the adversary.

Formally, consider any initial configuration at time $t_0$, such that $U(t_0)$ and $D(t_0)$ are respectively the multiset of positions and destinations of correct robots at time $t_0$.
Define $UD(t_0)$ to be the union of $U(t_0)$ and $D(t_0)$.
By considering the cycles started by correct robots after $t_0$, the following property holds:

\begin{lem}
\label{lem:dist-bords}
For each correct robot $i$ that starts a cycle after $t_0$,
the following inequalities hold:
$$D_i(t) \in [\dfrac{U_i(t)+Min(UD(t_0))}{2}, \dfrac{U_i(t)+Max(UD(t_0))}{2}]$$.
\end{lem}

\begin{proof}
The proof is twofold. First, we show that (1) $D_i(t) \geq (U_i(t)+Min(UD(t_0))) / 2$. Then, we prove the symmetric property (2) $D_i(t) \leq (U_i(t)+Max(UD(t_0))) / 2$.

\begin{enumerate}
\item $D_i(t) \geq (U_i(t)+Min(UD(t_0))) / 2$ :

Assume towards contradiction that for some robot $i$ that start a cycle at time $t_1\geq t_0$, there exists a time $t\geq t_1$ in this cycle such that:
\[ D_i(t) < \dfrac{U_i(t) + min(UD(t_0))}{2} \]

Note that $U_i(t_1) \geq U_i(t)$ because if robot $i$ moves between $t_1$ and $t$, it becomes closer to its destination $D_i(t)$.
Thus:
\[ D_i(t) < \dfrac{U_i(t_1) + min(UD(t_0))}{2} \ldots (1) \] 

This means that $distance( min(UD(t_0)) , D_i(t)) < distance( D_i(t), U_1(t))$. Denote by $d$ the distance between $U_i(t_1)$ and $D_i(t)$. Note that $D_i(t) < U_i(t_1)$.\\

The computation of $D_i(t)$ by $i$ is based on the configuration of
the network as last seen by robot $i$. That is, 
the configuration of the system at the beginning of its cycle $P(t_1)$. This implies that:

\[ D_i(t) = center ( trim_{2f}^i (P(t_1))) \ldots (2) \] 

We prove that (1) and (2) combined lead to a contradiction.

The location dependency property of the trimming function implies that $U_i(t_1) \in trim_{2f}^i (P(t_1))$.

So, up to this point we proved that there exists a point $U_i(t_1) \in
trim_{2f}^i (P(t_1))$ 
such that $U_i(t_1)>D_i(t)$ and $distance(U_i(t_1), D_i(t))=d$.

But since by (2), $D_i(t)$ is the center of $trim_{2f}^i (P(t_1))$, 
there must exists another point $q \in trim_{2f}^i (P(t_1))$, 
such that $q < D_i(t)$ and $distance(q, D_i(t))=d$. 

But we observed from (1) that $distance(min(UD(t_0)), D_i(t_1))<d$,
which implies that 
$distance(min(UD(t_0)), D_i(t_1))<distance(q, D_i(t))$. This means
that $q<min(UD(t_0))$. 
But $q \in trim_{2f}^i (P(t_1))$, so $min(trim_{2f}^i (P(t_1))) < min(UD(t_0))$. 
This contradicts lemma \ref{trim_sub_range}, which proves the first part of our lemma.

\item \textbf{(2)} $D_i(t) \leq (U_i(t)+Max(UD(t_0))) / 2$ :
The property is symmetric to (1) and can be proved using the same argument.

\end{enumerate}
\end{proof}


Let $S$ be a subset of correct robots,  
and define $UD_S(t)$ to be the multiset of their positions and destinations at time $t$. 

\begin{lem}
\label{lem:Size_m_2f}
If $|S| \geq m-2f$ and there exists a time $t_1 \geq t_0$ such that for each $t>t_1$, $max(UD_S(t)) \leq max(UD(t_0))-b$, 
then all computed destinations by all correct robots in cycles that start after $t_1$ are $\leq max(UD(t_0))-b/2$.
\end{lem}

\begin{proof}
Let $i$ be any correct robot that computes its destination $D_i$ in a cycle started after $t_1$, say at $t$.
We prove in the following that $D_i < max(UD(t_0)) - b/2$:\\

First, observe that since $max(UD_S(t)) < max(UD(t_0))-b$ 
and $|S| \geq m-2f > 2f$, 
then 
\[min(Trim_{2f}^i(P(t)) < max(UD(t_0))-b\]. 

Otherwise, $min(Trim_{2f}^i(P(t))$ would be greater than all the positions in $S$ ($>2f$ positions), 
which contradicts the definition of $Trim_{2f}^i$ (at most the $2f$ smallest positions are removed).

According to lemma~\ref{trim_sub_range}, we have 
\[max(Trim_{2f}^i(P(t)) < max(UD(t_0))\]

But $D_i$ is the center of $trim_{2f}^i(P(t))$, which means that
$distance(D_i, min(trim_{2f}^i(P(t))))$ 
must be equal to $distance(D_i, max(trim_{2f}^i(P(t))))$. 
Hence, 

\[D_i < max(UD(t_0)) - b/2 \]

\end{proof}


\begin{lem}
\label{lem:dest_inside_S}
If $|S| \geq m-f$ and at some time $t_1 \geq t_0$, $max(UD_S(t_1)) \leq max(UD(t_0))-b$, 
then all computed destinations by all correct robots in cycles that
start after $t_1$ are less or equal to $max(UD(t_0))-b/2$.
\end{lem}

\begin{proof}
First, we prove that after $t_1$, the robots in $S$ remains always at positions $< max(UD(t_0))-b$, meaning that all their computed destinations after $t_1$ are $< max(UD(t_0))-b$.

Assume the contrary: Let $i$ be the first robot in $S$ that starts a cycle after $t_1$ such that its computed destination in this cycle is $> max(UD(t_0))-b$.
This implies that $max (trim_{2f}^i(P (t_1)))> max(UD(t_0))-b$, 
which means that at least $2f+1$ positions in the network at $t_1$ are strictly greater than $max(UD(t_0))-b$.

If we add to these $2f+1$ positions that are greater than
$max(UD(t_0))-b$, the $m-f \geq n-2f$ positions in $S$ that are less
or equal than $max(UD(t_0))-b$, we get a total number of robots in the network that is strictly greater than $n$, which leads to contradiction.
This proves that all positions and destinations of robots in $S$ after
$t_1$ are less than or equal to $max(UD(t_0))-b$. 
Thus by lemma \ref{lem:Size_m_2f}, the destinations computed by
\textit{all} correct robots in the network are less than or equal to $max(UD(t_0))-b/2$. 

\end{proof}


The next Lemma states that if some computed destination is located in
the neighborhood of one extreme of the network, then a majority of
correct robots (at least $m-2f$) are located in the neighborhood of this extreme.

\begin{lem}
\label{lem-majority}
Let $D_i$ be a destination point computed by a correct robot $i$ in a cycle started at time $t$. 
If $D_i<min(UD(t))+b$, then at least $m-2f$ correct robots are located at positions that are $<min(UD(t))+2b$ at $t$.
\end{lem}

\begin{proof}
The computation of $D_i$ is based on the configuration of the network as last seen by robot $i$, 
that is the configuration at the beginning of the cycle at $t$, $P(t)$. 
So we first prove that at $t$, $max(trim_{2f}^i(P(t))) < min(UD(t)) + 2b$:

By hypothesis, $D_i < min(UD(t))+b$. But according to lemma ~\ref{trim_sub_range} , $min(UD(t)) \leq min(trim_{2f}^i(P(t)))$. Thus, 
$D_i < min(trim_{2f}^i(P(t)))+b$. This means that
\[distance(D_i, min(trim_{2f}^i(P(t)))) < b \]

But $D_i$ is the center of $trim_{2f}^i(P(t))$, which means that
$distance(D_i, min(trim_{2f}^i(P(t))))$ 
must be equal to $distance(D_i, max(trim_{2f}^i(P(t))))$. 
Hence, 
\[ max(trim_{2f}^i(P(t)))) < D_i + b \]

But since by hypothesis $D_i <min(UD(t))+b$, we have 
\[max(trim_{2f}^i(P(t)))) < min(UD(t)) + 2b\]

This means that at most $2f$ positions (which may be correct) are $\geq min(UD(t)) + 2b$ at $t$. This completes the proof. 

\end{proof}


Let $U(t_0)$ and $D(t_0)$ be respectively the multisets of positions and destinations of correct robots at the initial time $t_0$, and define $UD(t_0)$ to be the union of $U(t_0)$ and $D(t_0)$. Take $b$ to be any distance $< diameter(UD(t_0))/4$, for example $b=diameter(UD(t_0))/10$. 

The next lemma states that if a correct robot elected at $t>t_0$ is located inside the range $(min(UD(t_0))+b, max(UD(t_0))-b)$, then the destinations points computed by correct robots after $t$ are either \textit{all} $\leq max(UD(t_0))-b/4$ or \textit{all} $\geq min(UD(t_0))+b/4$. This means that the election of a robot located inside $(min(UD(t_0))+b, max(UD(t_0))-b)$ is a sufficient condition to convergence.

\begin{lem}
\label{lem:dest_in_range}
Let $t_1$ be the first time at which all correct robots in the network executed a complete cycle at least once since $t_0$.

If a correct robot is elected at $t>t_1$ and is located inside $[min(UD(t_0))+b, max(UD(t_0))-b]$, then the destination points computed by correct robots in cycles that start after $t$ are either all located at positions $\leq max(UD(t_0))-b/4$ or all located at positions $\geq min(UD(t_0))+b/4$. 

\end{lem}

\begin{proof}


Let $i$ be a correct robot that is elected at time $t>t_1$ 
and whose position $U_i(t)$ is inside $[min(UD(t_0))+b, max(UD(t_0))-b]$.
According to lemma \ref{lem:i_inf_f+1}, either $U_i(t) \geq U_{m-f}(t)$ or $U_i(t) \leq U_{f+1}(t)$. 
Thus, we separate the analysis into two cases depending on the rank of the elected robot:

\begin{itemize}
	\item \textbf{Case 1: } $U_i(t) \geq U_{m-f}(t)$.

Define $S(t)$ to be the set of correct positions $\{U_1(t), ..., U_{m-f}(t), ..., U_i(t)\}$, 
and note that $|S(t)| \geq m-f$.

By hypothesis, $U_i(t) \leq max(UD(t))-b$ 
which implies that the positions of all robots in $S(t)$ are $ \leq min(UD(t))+b$. 
Thus by lemma \ref{lem:dist-bords}, the destinations of all robots in $S(t)$ are $ \leq max(UD(t_0))- b/2$. 
This means that $range_{S}(t)$, 
the range of positions and destinations of robots in $S(t)$ is such that at $t$, $max(range_{S}(t)) \leq max(UD(t_0))-b/2$. 
Hence, according to lemma \ref{lem:dest_inside_S},
all destinations points computed by correct robots in cycles that start after $t$ are $\leq max(UD(t_0))-b/4$.
\\
	\item \textbf{Case 2: } $U_i(t) \leq U_{f+1}(t)$.

The case is symmetric and we prove by a similar argument to \textbf{Case 1} that all destination points computed by correct robots are $\geq min(UD(t_0))+b/4$, 
which proves our lemma.
\end{itemize}
\end{proof}


\begin{lem}
\label{lem:shrinking}
Algorithm 1 is shrinking in CORDA model under a fully asynchronous scheduler when $n>5f$.
\end{lem}

\begin{proof}

Let $U(t_0)=\{U_1(t_0), U_m(t_0)\}$ be the configuration of correct robots at initial time $t_0$, 
and let $D(t_0)=\{D_1(t_0), ..., D_m(t_0)\}$ the multiset of their destinations at $t_0$.
Define $UD(t_0)$ to be the union of $U(t_0)$ and $D(t_0)$, 
and let $diam(t_0)$, the diameter at $t_0$, be equal to $max(UD(t_0)) - min(UD(t_0))$. 
$U(t)$, $D(t)$, $UD(t)$ and $diam(t)$ for each $t>t_0$ are defined similarly.

Let $t_1$ be the first time at which every correct robot in the network has executed a whole cycle at least once since $t_0$.
We consider the evolution of the network after $t_1$. 
The aim of this is to apply lemma \ref{lem:dist-bords}, that is, based only on the position of a correct robot, we can give bounds on its destination point which is especially interesting in the case of a robot executing a Move phase of its cycle.

We take into account all the computed destinations by correct robots after $t_1$ 
and we distinguish between two cases: 
(1) the case when all destinations computed after $t_1$ are inside $[min(UD(t_0))+\dfrac{diam(t_0)}{10}, max(UD(t_0))-\dfrac{diam(t_0)}{10}]$.
and (2) the case when a computed destination after $t_1$ lay outside this range.
We show that in both cases, there is a time at which the diameter of correct positions decreases by a factor of at least $39/40$.

\begin{itemize}

\item \textbf{Case 1:} All destinations computed by correct robots in cycles started after $t_1$ are inside the range

\begin{center}
 $[min(UD(t_0)) + diam(t_0)/10, max(UD(t_0))-diam(t_0)/10]$.
\end{center}

In this case, 
since each robot $i$ is guaranteed to move a minimal distance of $\delta_i$ before it can be stopped by the adversary, 
there is a time $t_2 \geq t_1$ when all correct robots are located inside $[min(UD(t_0))+diam(t_0)/10, max(UD(t_0))-diam(t_0)/10]$.
Thus $diam(t_2)=diam(t_0)* 4/5$, and by setting $\alpha=4/5$, our algorithm is shrinking.

\item \textbf{Case 2:} There is a destination $D_i$, computed by a correct robot $i$ in a cycle started after $t_1$, that is outside the range 

\begin{center}
 $[min(UD(t_0)) + diam(t_0)/10, max(UD(t_0))-diam(t_0)/10]$.
\end{center}

This means that either $D_i<min(UD(t_0)) + diam(t_0)/10$ or $D_i>max(UD(t_0))-diam(t_0)/10$. 
Since the two cases are symmetric, there is no loss of generality to assume that $D_i<min(UD(t_0)) + diam(t_0)/10$. 

The calculation of $D_i$ is based on the configuration of the network as seen by robot $i$ at the beginning of the cycle, say at $t_2$ (with $t_2 \geq t_1$).
Thus, according to lemma \ref{lem-majority}, at $t_2$, 
at least $m-2f$ correct robots are located at positions $<min(UD(t_0)) + diam(t_0)/5$. 
Denote by $S(t_2)$ the set of these robots.
By lemma \ref{lem:dist-bords}, the destinations of robots in $S(t_2)$ are $< min(UD(t_0)) + diam(t_0)*(3/5)$.
Thus, the positions and destinations of robots in $S(t_2)$ are $<max(UD(t_0)) - diam(t_0)*2/5$.

We now observe the positions of elected robots whose rank is $ \leq f+1$ and which are activated after $t_2$.
We separate the analysis into two subcases: 

	\begin{itemize}

	\item \textbf{Subcase 2A:} There is a time $t>t_2$ at which is elected a correct robot $i$ whose rank is $ \leq f+1$ and whose position $U_i(t)$ is $>min(UD(t_0)) + diam(t_0)/10$. 
Notice that since $|S(t_2)|>m-2f$, 
$U_i(t)$ is also $<max(UD(t_0)) - diam(t_0)*2/5$ which is the upper bound on the positions of robots in $S(t_2)$. 
Thus, $U_i(t) \in [min(UD(t_0))+diam(t_0)/10, max(UD(t_0))-diam(t_0)/10]$ and according to lemma \ref{lem:dest_in_range}, the diameter eventually decreases by a multiplicative factor of $1-1/40$. Hence, by setting $\alpha=39/40$ the lemma follows.

	\item \textbf{Subcase 2B:} All elected correct robots that are activated after $t_2$ and whose rank is $ \leq f+1$ are located at positions  $< min(UD(t_0)) + diam(t_0)/10$. 
This implies, according to lemma \ref{lem:dist-bords}, that the positions of these elected robots remain always at positions $<max(UD(t_0)) - diam(t_0)*9/20$. 
Thus, all robots in $S(t_2)$ remain always at positions $<max(UD(t_0)) - diam(t_0)*9/20$ ($\forall t>t_2$).
	
	According to lemma \ref{lem:Size_m_2f}, all destinations computed at cycle that start after $t_2$ are $<max(UD(t_0)) - diam(t_0)*9/40$. And since robots are guaranteed to move toward destinations by a minimum distance before they can be stopped by the adversary, they all end up located at positions $<max(UD(t_0)) - diam(t_0)*9/40$. Hence there is a time $t>t_2$ such that $diam(t)=diam(t_0)*(1-9/40)$. It suffices to set $\alpha=31/40$ and the lemma follows.
	
 	\end{itemize}
\end{itemize}
Consequently, we set $\alpha=39/40$ and the lemma is proved.
\end{proof}

\end{document}